\documentclass[11pt]{article}
\usepackage[utf8x]{inputenc}
\usepackage[english]{babel}
\usepackage{amsmath}
\usepackage{amsfonts}
\usepackage{amssymb}
\usepackage{latexsym}
\usepackage{graphicx}
\usepackage{psfrag}
\usepackage{longtable}
\usepackage{mathrsfs}
\usepackage{enumitem}
\usepackage{cancel}

\newenvironment{changemargin}[2]{\begin{list}{}{%
\setlength{\topsep}{0pt}%
\setlength{\leftmargin}{0pt}%
\setlength{\rightmargin}{0pt}%
\setlength{\listparindent}{\parindent}%
\setlength{\itemindent}{\parindent}%
\setlength{\parsep}{0pt plus 1pt}%
\addtolength{\leftmargin}{#1}%
\addtolength{\rightmargin}{#2}%
}\item }{\end{list}}

\newtheorem{theorem}{Theorem}
\newtheorem{lemma}{Lemma}

\newtheorem{corollary}{Corollary}
\newtheorem{prop}{Proposition}
\newtheorem{rema}{Remark}

\newenvironment{proof}{\noindent \emph{Proof.}\ }{\hfill
    $\Box$\vspace{1em}}

\def\R{\textrm{I\kern-0.21emR}}

\title{Exact values for three domination-like problems in circular and infinite grid graphs of small height}
 
\author{M. Bouznif\thanks{a-SIS, 8, rue de la Richelandi\`ere, F-42100 Saint \'Etienne, France} \and J. Darlay\thanks{Innovation24, LocalSolver, 36 avenue Hoche, 75008 Paris, France} \and J. Moncel\thanks{LAAS-CNRS, Universit\'e de Toulouse, CNRS, Universit\'e Toulouse 1 Capitole, IUT de Rodez, Toulouse, France} \thanks{F\'ed\'eration de recherche Maths \`a Modeler, 100 rue des maths, St Martin d'H\`eres, France} \and M. Preissmann \thanks{Univ. Grenoble Alpes, CNRS, G-SCOP, 38 000 Grenoble, France.}} 
\date{}
\begin{document}
\maketitle

\begin{abstract}
In this paper we study three domination-like problems, namely identifying codes, locating-dominating codes, and locating-total-domi\-na\-ting codes. We are interested in finding the minimum cardinality of such codes in circular and infinite grid graphs of given height. We provide an alternate proof for already known results, as well as new results. These were obtained by a computer search based on a generic framework, that we developed earlier, for the search of a minimum labeling satisfying a pseudo-$d$-local property in rotagraphs. 
\end{abstract}

\section{Framework}

The aim of this paper is to determine  particular subsets of vertices of minimum density in grid graphs of fixed height. All these subsets are dominating sets with special properties that are related to several applications such as fault diagnosis in array of processors \cite {KCL98} or safeguard analysis of a facility using sensor networks \cite{Slater87}. We show here that the corresponding problems are relevant from the method described in \cite{BMP} and provide new results for grids of small heights (at most $4$). 

This section contains all basic definitions and a brief bibliographic review of the subject. 

The next section is dedicated to the description of the adaptation of the theoretical framework of \cite{BMP} to the case of the search of the minimum cardinality of an $ID$-code in circular strips of given height. 

Then, in Section~\ref{constant sec}, we will explain why the method works in constant time. We will also describe how one can find the minimum cardinality of  $ID$-codes in non-circular strips as well as the minimum densities of $ID$-codes in infinite strips.

In Section~\ref{sec:algo} we provide details related to the implementation of the algorithms. There, the reader will find information such as technical tricks, memory used in the RAM, or running times. 

Our results concerning the minimum cardinality or density of $ID$-, $LD$-, and $LTD$-codes in finite circular and in infinite strips of the square, triangular, and king grids, are displayed in Section~\ref{sec:results}.

\subsection{Graphs and codes}

A \textit{graph} $G$ is a couple $(V,E)$ in which $V$ is a set of \textit{vertices} and $E$ is a set of 2-elements subsets of $V$ called \textit{edges}. Two vertices that are joined by an edge in $G$ are said to be \textit{neighbors}. For a vertex $v\in V$, the set of neighbors of $v$ in $G$ is denoted by {\boldmath $N_G(v)$}, and the closed neighborhood $N_G(v) \cup \{v\}$ of $v$ is denoted by {\boldmath $N_G[v]$} (the subscript $G$ may be omitted when there is no ambiguity).

A \textit{code} of a graph $G$ is simply a subset of vertices of $G$.

Given a code $C$ of a graph $G$, we say that 
	a vertex $v$ is \textit{dominated by $C$} if $N_G[v] \cap C \neq \emptyset$,
	and it is \textit{totally dominated by $C$} if $N_G(v) \cap C \neq \emptyset$.
	Two distinct vertices $u$ and $v$ are said \textit{separated by $C$} if $N_G[u] \cap C \neq N_G[v] \cap C$.

A code $C$ of a graph $G$ is  :
\textit{dominating}, or a \textit{$D$-code},  if every vertex of $G$ is dominated by $C$~;
\textit{total-dominating}, or a \textit{$TD$-code}, if every vertex of $G$ is totally dominated by $C$~;
 \textit{locating-dominating}, or an \textit{$LD$-code}, if it is a $D$-code and every  two distinct vertices $u$ and $v$ not in $C$ are separated by $C$~;
 \textit{locating-total-dominating}, or an \textit{$LTD$-code}, if it is a $TD$-code and an $LD$-code~; 
 \textit{identifying}, or an \textit{$ID$-code}, if it is a $D$-code and every two distinct vertices $u$ and $v$ of $G$ are separated by $C$.

From previous definitions it is immediate to see that given a graph $G$ and a code $C$ of $G$, the following holds:
	\begin{description}
		\item[-]  $C$ is an $ID$-code of $G$ $\Rightarrow$ $C$ is a $LD$-code of $G$,
		\item[-]  $C$ is an $LTD$-code of $G$ $\Rightarrow$ $C$ is an $LD$-code and a $TD$-code of $G$,
		\item[-]  $C$ is an $LD$-code or a $TD$-code of $G$ $\Rightarrow$ $C$ is a $D$-code of $G$,
			\end{description}

For the notions of $D$, $TD$, $LD$ and $LTD$-codes, see~\cite{HHS98, HY13}. For the notion of $ID$-code, see~\cite{KCL98}.

The following Lemma holds by the definition of domination and separation.

\begin{lemma}\label{lem:superset}
Let $C$ be a code of a graph $G=(V,E)$ and $v \in V \setminus C$. In $G$, every vertex dominated by $C$ is dominated by $C\cup \{v\}$ and every two vertices separated by $C$ are separated  by $C\cup \{v\}$.
\end{lemma}
As a corollary of this Lemma we get that $G$ contains a $D$-code (resp. a $TD$-, an $LD$-, an $LTD$-, an $ID$-code) only if $V$ itself is a $D$-code (resp. a $TD$-, an $LD$-, an $LTD$-, an $ID$-code). Hence, deciding if in a given graph there exists one of these codes is not hard. The problem is to find one of minimum cardinality.

\subsection{Grids and strips}

We define three infinite graphs, that all have $\mathbb{Z}^2$ as vertex set :

The \textit{square grid}, denoted {\boldmath $\mathcal{S}$}, is the graph such that $(i,j)(k,l)$ is an edge whenever $|i-k|+|j-l|=1$ (see Figure \ref{bandecarrecode}). The \textit{triangular grid}, denoted {\boldmath $\mathcal{T}$}, is the graph obtained by adding to {\boldmath $\mathcal{S}$} all edges $(i,j)(k,l)$ such that $(k-i)=(j-l)=1$ (see Figure \ref{bandeautrecode}). The \textit{king grid}, denoted {\boldmath $\mathcal{K}$}, is the graph obtained by adding to {\boldmath $\mathcal{T}$} all edges $(i,j)(k,l)$ such that $(k-i)=(l-j)=1$ (see Figure \ref{bandeautrecode}).

Any of these three graphs will be said to be a \textit{grid}.

Consider a grid $\mathcal G$ and a positive integer $h$. 

The \textit{infinite strip} of height $h$ of $\mathcal G$, denoted {\boldmath $\mathcal G_h$}, is the subgraph of $\mathcal G$ induced by the vertices $(i,j)$ with $i\in\{1,\ldots,h\}$. The \textit{infinite toroidal strip} of height $h \ge 3$, denoted {\boldmath $\mathcal{S}_{\circ h }$}, is obtained from the infinite square strip $\mathcal{S}_h$ by adding all edges $(1,j)(h,j)$ for $j\in\mathbb{Z}$.

The \textit{finite strip} of height $h$ and size $s\geq1$ of $\mathcal{G}$ is the subgraph of $\mathcal{G}$, denoted {\boldmath $\mathcal{G}_{h,s}$} induced by the vertices $(i,j)$ with $i\in\{1,\ldots,h\}$ and $j\in\{1,\ldots,s\}$. The \textit{circular strip} of $\mathcal{G}$ of height $h\geq1$ and size $s\geq3$, denoted {\boldmath $\mathcal{G}_{h,s}^{\circ}$}, is obtained from $\mathcal{G}_{h,s}$ by adding all edges $(i,s)(i',1)$ such that $(i,s)(i',s+1)$ is an edge of {\boldmath $\mathcal{G}$}, for $i,i'\in\{1,\ldots,h\}$. The \textit{toroidal circular strip} of height $h\geq 3$ and size $s\geq3$, denoted {\boldmath $\mathcal{S}^{\circ}_{\circ h,s}$}, is obtained from the circular square strip $\mathcal{S}_{h,s}^{\circ}$ by adding all edges $(1,j)(h,j)$ for $j\in\{1,\ldots,s\}$.

Let $G$ be a graph whose set of vertices is included in $\mathbb{Z}^2$. The $k$-th \textit{column} of $G$ denotes the set of vertices $(i,j)$ of $G$ such that $j=k$. A set $\mathcal E$ of columns of $G$ is said to be a set of  \textit{consecutive columns} if there exist integers $k$ and $l$,  $k\le l$, such that $\mathcal E$ is equal to the set of $i$-th columns for $i \in [k,l]$. The columns that are \textit{neighbors of the $k$-th column} of $G$ are the $(k-1)$-th and the $(k+1)$-th columns (if defined), with addition modulo $n$ in the case $G$ is a circular strip on $n$ columns. Notice that if $G$ is a subgraph of a grid then any of its vertices has neighbors only in its column or in a neighbor of it.
Given a strip $S$ (of any kind), any non-circular strip induced in $S$ by a set of consecutive columns of $S$ will be a called a \textit{substrip} of $S$. In the case $S$ is a circular strip of size $s$, a substrip of $S$ of size $s$ is obtained from $S$ by deleting the edges between a pair of consecutives columns of $S$.

\subsection{Literature review}

There is a broad literature about $LD$- and $ID$-codes in infinite grids, see for instance~\cite{BHL05, CHHL01, CHHL04, P12}. $ID$- and $LD$-codes in infinite strips were addressed in~\cite{BCHL04, DGM04, JM05}. As for 
$LTD$-codes in infinite strips, the problem was studied in~\cite{HR12, J}.

In the above-mentioned references, bounds or exact values are given for the minimum density of a code in infinite grids or strips. These results are obtained by combinatorial arguments based on the analysis of local configurations.

\medskip

There are also papers dealing with the algorithmic aspects of finding the minimum cardinality of some codes in grids or in grid-like structures. For instance, in~\cite{DGM04,HHH86,LS94,S98,vW07,Z99}, efficient algorithms are provided to compute the minimum cardinality of a $D$- or $ID$-code in broad classes of graphs, containing in particular circular strips. The classes of graphs involved in these papers are called fasciagraphs and rotagraphs ; fasciagraphs generalize strips of grids and rotagraphs generalize circular strips of grids. For short, a fasciagraph is constituted by multiple consecutive copies of a given graph, each copy being linked to the next one by a fixed scheme. 

In~\cite{BMP} one can find the definition of a fasciagraph and of a rotagraph, as well as a general framework that unifies the results presented in the above-mentioned papers. 
It is shown there that, due to the repetitive structure of these graphs,  dynamic programming can  be applied to address  optimization problems that are --- in some sense --- ``local''.

In Section~\ref{sec:results} we present the results we obtained by implementing the algorithm described in~\cite{BMP}, to get new results on the minimum cardinality of $ID$-, $LD$-, and $LTD$-codes in strips of small height.

\medskip

\section{The algorithm}

The  present section is dedicated to the description of the adaptation of the theoretical framework of \cite{BMP} to the case of the search of an $ID$-code of minimum cardinality in circular strips of given height.
The algorithms for $LD$- or $LTD$-codes being similar, they are not described hereafter. Notice that in \cite{BMP} the algorithm for finding a minimum $D$-code is described, and it can easily be adapted for a $TD$-code.

\subsection{Labelings, codes and pseudo-$d$-local properties}

In the framework of~\cite{BMP}, we can address combinatorial problems whose solutions may be described as particular $q$-labelings of the vertices of an associated graph. Given an integer $q\geq2$, a $q$-labeling of a graph is simply a function $f$ that maps each vertex $v$ of the graph to an integer $f(v)\in\{0,\ldots,q-1\}$. There is a one-to-one correspondence between $2$-labelings and codes by considering that the vertices of the code are exactly those that are labeled $1$. Given a $2$-labeling $f$ of the vertices of a graph we will denote by $C_f$ the corresponding code.  We will see now that, for the kind of $2$-labelings of circular strips we are looking for, it is enough to require a property of the labeling limited to "small" subgraphs of the strip. From now on, we will  focus on the special case of a minimum $ID$-code.

Given a labeling $f$ of a strip $\mathcal{G}_{h,s}$ and $1\le i \le j \le s$, {\boldmath $f_{i,j}$} denotes the labeling of $\mathcal{G}_{h,j-i+1}$ corresponding to the restriction of $f$ to the columns of $\mathcal{G}_{h,s}$ numbered from $i$ to $j$. Similarly, for a  labeling $f$ of a circular strip $\mathcal{G}_{h,s}^{\circ}$ and two integers $1\le i,j \le s$ we will denote by {\boldmath $f_{i,j}$} the labeling by $f$ of the columns $i$, $i+1\ldots, j-1, j$ of $\mathcal{G}_{h,s}^{\circ}$  (addition modulo $s$).

Let us now see more precisely in which sense we consider the property of being an $ID$-code as "local". We will say that a $2$-labeling $f$ of a strip $\mathcal{G}_{h,s}$ or a circular strip $\mathcal{G}_{h,s}^{\circ}$, of size $s \ge 5$ (for some grid $\mathcal{G}$) satisfies the property {\boldmath$\mathcal P^{I}$} if in every  (non circular) substrip $F$ of size $5$, the vertices in the three middle columns are dominated and separated from each other by the vertices of $C_f$ that are in $F$. It is easy to see that one can check within a finite number of steps if a labeling of a finite circular strip of size at least $5$ satisfies this property. From the following theorem we can then deduce that being an identifying code of a circular string is a pseudo-5-local property (as defined in \cite{BMP}).

\begin{theorem} \label{five}  The code $C_f$ associated to a $2$-labeling $f$ of a circular strip $\mathcal{G}_{h,s}^{\circ}$ ($s \ge 5$)  is an $ID$-code of $\mathcal{G}_{h,s}^{\circ}$ if and only $f$ satisfies $\mathcal P^{I}$.                                                                                                                                                                                                                                                   
\end{theorem}

\begin{proof}
We remark that the vertices that are in the three middle columns of a $5$-columns substrip $F$ of $\mathcal{G}_{h,s}^{\circ}$ have their neighborhoods included in $F$, so the condition is clearly necessary.  Assume now that the condition is fulfilled and let us consider any vertex $v$ of $\mathcal{G}_{h,s}^{\circ}$. It belongs to the third column of a substrip $F$ of five consecutive columns of $\mathcal{G}_h^{\circ s}$ so, since $f$ satisfies $\mathcal P^{I}$, $v$ is dominated already in $F$. Hence $C_f$ is dominating. Let now $w$ be a vertex of $\mathcal{G}_{h,s}^{\circ}$ distinct from $v$. If $v$ and $w$ are not contained in the set of vertices of three consecutive columns of $\mathcal{G}_{h,s}^{\circ}$ then their closed neighborhoods are disjoint and have a non-empty intersection with $C_f$ (since $C_f$ is dominating), so $v$ and $w$ are separated by $C_f$. Assume now that $v$ and $w$ are contained in the set of vertices of three consecutive columns of $\mathcal{G}_{h,s}^{\circ}$. Since $s \ge 5$, these columns are the three middle columns of a substrip  of $\mathcal{G}_{h,s}^{\circ}$ of size $5$. The fact that $f$ satisfies $\mathcal P^{I}$ entails that $v$ and $w$ are separated by $C_f$.
\end{proof}

We notice that, if we were interested by a dominating code of a circular strip, it would have been sufficient to verify that in every substrip $F'$ of size~$3$, the vertices in the middle column are dominated by the vertices of $F'$ labeled $1$ \cite{BMP} ; thus this property is pseudo-3-local. 

For all other kind of codes introduced in the present paper, we can define an associated property $\mathcal P^{loc}$ of substrips of size $d$ (for some fixed integer~$d$), similar to $\mathcal P^{I}$, and a theorem similar to Theorem \ref{five} holds. 
As it will be seen below, this enables us to find a minimum code of any kind by considering paths of minimum weight in an associated directed graph whose vertices are, basically, substrips of size $d-1$ equipped with appropriate vertex labelings, and whose arcs correspond to substrips of size $d$ satisfying $\mathcal P^{loc}$, the weight of an arc being equal to the cardinality of the associated code in the last fiber.

\subsection{Computation of a minimum $ID$-code in a circular strip}

In the rest of this section we will assume we are given a grid $\mathcal G$ and a height~$h$.
Our algorithm to compute the minimum cardinality of an $ID$-code of a circular strip of $\mathcal G$ of height $h$ needs to build an auxiliary directed graph with a length function on the arcs.

Before to describe this graph we need some definitions and notation.

A {\it directed graph} $\vec{G}$ is a couple  $(V,A)$, where
{\it $V$} is a set of elements called {\it vertices} and {\it $A$} is
a subset of couples of elements of $V$ called {\it arcs}. An arc $(u,u)$ is called a {\it loop}.

Let $k$ be a positive integer and $\vec{G}=(V,A)$ be a directed graph. 

A {\bf \it path $P$ of cardinality $k$} of $\vec{G}$, also called {\it$k$-path},   is a
sequence  $v_1,\ldots,v_{k+1}$ of (non necessarily distinct) vertices such that $(v_i,v_{i+1}) \in
A$ for all $i \in \{1,\ldots,k\}$. We then say that $P$ is a path from 
$v_{1}$ to $v_{k+1}$.

A {\it  circuit $C$ of cardinality $k$}, also called {\it $k$-circuit}, of a directed graph $\vec{G}=(V,A)$ ($k \ge 1$), is a
path $v_1,\ldots,v_{k+1}$ such that $v_1=v_{k+1}$. If $\{v_1,v_{k+1}\}$ is the only pair of non-distinct vertices in the sequence $v_1,\ldots,v_{k+1}$, the circuit is said to be {\it elementary}.
The {\it cardinality of a path (or a circuit) $Q$} is denoted by
{\it{$|Q|$}}. 

A strongly connected component of a directed graph $\vec{G}$ is a maximal 
subgraph of $\vec{G}$ such that there exists a path from any vertex to any other 
vertex. A strongly connected component of a directed graph $\vec{G}$ 
is said to be {\it trivial} if it contains only one vertex and no arc.

Let $\vec{G}=(V,A)$ be a directed graph. If a length function $\ell : A \rightarrow \mathbb{N}$ is given, then we say that $\vec{G}$ is an $\ell$-graph, and we define the {\it length of a path} $P=v_1,\ldots,v_{k+1}$ of $\vec{G}$ 
($k\ge 1$) as follows:

$$\ell(P) = \ell(v_1,v_2) + \ell(v_2,v_3) + \ldots + \ell(v_{k},v_{k+1}).$$

The {\it mean} of a $k$-path or a $k$-circuit $C$ of an $\ell$-graph $\vec{G}$ is $mean(C)= \frac{\ell(C)}{k}$, the mean length of an edge in $C$. Assume that $\vec{G}$ has a finite number of vertices. The {\it minimum mean of a circuit in $\vec{G}$}
is denoted by $\lambda(\vec{G})$ ; in case $\vec{G}$ has no circuit, $\lambda(\vec{G})$ is set to $\infty$. Notice that, since the mean of a circuit cannot be lower than the minimum mean of its elementary subcircuits, $\lambda(\vec{G})$ is equal to the minimum mean of an elementary circuit of $G$. As the number of elementary circuits of a finite graph is finite,  $\lambda(\vec{G})$ is well-defined for a finite graph. 
We call {\it min-mean component of $\vec{G}$}, any non trivial strongly connected component of the subgraph of $\vec{G}$
induced by arcs belonging to circuits of mean equal to $\lambda(\vec{G})$. The {\it periodicity of a min-mean component of $\vec{G}$} is the gcd of the cardinalities of its elementary circuits of mean  $\lambda(\vec{G})$.

Given two $2$-labelings $f$ and $f'$ of the $4$-columns strip $\mathcal{G}_{h,4}$, we will say that {\it $f$ is compatible with $f'$} if $f$ labels the
 last three columns of $\mathcal{G}_{h,4}$ exactly as $f'$ labels the first three columns of $\mathcal{G}_{h,4}$. Given two compatible labelings $f$ and $f'$, the concatenation of $f$ and $f'$
denoted by $f \triangleright f'$ is the $2$-labeling of the $5$-columns strip $\mathcal{G}_{h,5}$ which is equal to $f$ on the first four columns and to $f'$ on the last four columns. If furthermore the concatenation of $f$ and $f'$ 
 satisfies $\mathcal P^{I}$ we will say that $f$ is $I$-compatible with $f'$.

We denote by {\boldmath $\vec{\mathcal{G}}^{I}_h$} the directed $\ell$-graph defined as follows :
\begin{itemize}
	\item[-] the vertices of $\vec{\mathcal{G}}^{I}_h$ are all the $2$-labelings of the $4$-columns strip $\mathcal{G}_{h,4}$,
	\item[-] the arcs of $\vec{\mathcal{G}}^{I}_h$ are the couples $(u,v)$ of
vertices of $\vec{\mathcal{G}}^{I}_h$ (not necessarily distinct) such that $u$ is $I$-compatible with $v$,
	\item[-] the length $\ell(u,v)$ of  an arc $(u,v)$ of $\vec{\mathcal{G}}^{I}_h$ is the number of vertices of the 5-th column of $\mathcal{G}_{h,5}$ labeled $1$ by $u \triangleright v$.
	\end{itemize}
The graph $\vec{\mathcal{G}}^{I}_h$ will be called \textit{auxiliary graph} for $ID$-codes in strips of height~$h$.

From Corollary $2$  and Section 5 in \cite{BMP} that are valid for any pseudo-$d$-local property, we get the following "specific" theorem for $ID$-codes.

\begin{theorem}\cite{BMP} \label{chem}
 For every  integers $k\ge 5$ and $l \ge 1$, there exists an $ID$-code of $\mathcal{G}_{h,k}^{\circ}$ of cardinality $l$ if and only if $\vec{\mathcal{G}}^{I}_h$  contains a $k$-circuit of length $l$.
\end{theorem} 

From this theorem we immediately get that we may obtain the minimum cardinality of an $ID$-code of $\mathcal{G}_{h,k}^{\circ}$ by computing the minimum length of a $k$-circuit in $\vec{\mathcal{G}}^{I}_h$.
There is a well-known way to solve the problem of computing the minimum length of a $k$-circuit in a directed graph with a length function on the arcs. To describe it we need some additional definitions.

\noindent Given two $n \times n$ matrices $A,B$  with entries in $\mathbb{N} \cup \{\infty\}$, we define the
{\it product of} {\it{$A$}} {\it and} {\it{$B$}},
denoted {\it{$A B$}}, as the $n\times n$ matrix such
that: $[A B]_{i,j}=Min_{k=1}^n (A_{i,k} + B_{k,j})\mbox{
for all }i,j \in \{1,\ldots,n\}$. The product
${A A\ldots A}$ ($k$ occurrences of $A$) is
denoted by {\it $A^k$}.

Let $\vec{G}$ be a directed $\ell$-graph on $n$ vertices. Given a numbering $u_1,u_2,\ldots,u_n$ of the vertices of $\vec{G}$, the {\it length-matrix} of 
$\vec{G}$, 
is the $n \times n$ matrix $\Pi$  defined as follows:

 $ \Pi_{i,j}= \left\{ 
\begin{array}{ll}
	\infty \mbox{ if } (u_i,u_j) \notin A\\
	\ell(u_i,u_j) \mbox{ otherwise,}	 
\end{array} 
\right.$
for $i,j \in \{1,\ldots,n\}$.

The following result is well-known and very useful.

\begin{theorem}[Section 4.2 in \cite{GM02}] \label{lemparc}
Let $\Pi$ be the length-matrix of an $\ell$-graph $\vec{G}$ 
with vertex set $\{u_1,u_2,\ldots,u_n\}$. For any integers $k\geq 1$ and $i,j \in
\{1,\ldots,n\}$, we have: 

$[\Pi^k]_{i,j} = \left\{ \begin{array}{ll}
 \infty \mbox{ if there is no } k\mbox{-path from } u_i \mbox{ to }
u_j \mbox{ in } \vec{G},\\
 Min\{\ell(P) \vert P \mbox{ is a } k\mbox{-path from } u_i \mbox{ to }
u_j \mbox{ in } \vec{G}\} \mbox{, otherwise.}
\end{array}
\right.$

\end{theorem}

From all the results above, we can obtain the minimum cardinality of an $ID$-code of a circular strip $\mathcal{G}_{h,k}^{\circ}$ ($k \ge 5$), by  generating the directed $\ell$-graph $\vec{\mathcal{G}}^{I}_h$, computing the $k$-th power of the length-matrix of $\vec{\mathcal{G}}^{I}_h$ and then returning the smallest element in the diagonal of this matrix. For a fixed height $h$ this algorithm has a running-time which is polynomial in $k$.

In the next section we will see that the length-matrix of $\vec{\mathcal{G}}^{I}_h$ has another very interesting property that enables one to compute this minimum cardinality for any size, and even for infinite strips, in constant-time.

\section{Constant-time computation of minimum $ID$-codes in strips of given height} \label{constant sec}

In this Section we will define a stable matrix and express important results on stable matrices due to Moln\'arov\'a and Pribi\v{s} \cite{MP00}. These results are essential to show that the minimum cardinality of an $ID$, $LD$, or $LTD$-code in circular and non circular strips of given height and any size  may be computed by a constant-time algorithm. Furthermore, as a corollary of these results we get that the minimum density of a code in an infinite strip of height $h$ is the same as the minimum density of a code in a circular strip of height $h$.

\subsection{Stable matrices}
Given an $n \times n$ matrix $A$ with entries in $\mathbb{N} \cup \{\infty\}$ and an integer $c$, we define the {\it sum of} {\it{$A$}}
{\it and} {\it{$c$}}, denoted {\it{$A+ c$}}, as
the $n \times n$ matrix such that $[A+ c]_{i,j}=A_{i,j}+
c$ ($i,j \in \{1,\ldots,n\}$). We also say that $A+ c$ is the {\it translation} of $A$ by $c$.

A matrix $\Pi$ is said {\it $(c,p,u)$-stable} with
{\it transfer factor $c \in \mathbb N$},   {\it period $p \in \mathbb N$}, and {\it start 
$u\in \mathbb N$}, if 
 	$\Pi^{i+p}=\Pi^{i} + c, \forall i\geq u.$

\begin{rema}  \label{stabper}
 If $\Pi$ is $(c,p,u)$-stable, then the sequence of the powers of $\Pi$ is \textit{pseudo-periodic}, that is to say, it has the following
form: $$\Pi,\Pi^2,\ldots,\Pi^{u-1},[S_0],[S_0 + c], [S_0 +
2 c],\ldots,[S_0 + ic],\ldots$$
where $[S_0]$ is the sequence $\Pi^u,
\Pi^{u+1},\ldots,\Pi^{u+p-1}$, and for $j\geq 1$, $[S_0 +
jc]$ is the sequence $\Pi^{u}+ jc,\Pi^{u+1} + jc,\ldots,\Pi^{u+p-1}\ +jc$. So, once the first $u+p$ powers of $\Pi$ have been computed, for 
any integer $k> u+p$, 
$\Pi^k$ can be obtained by a constant number of elementary operations.
\end{rema}

A matrix is said {\it stable} if it is $(c,p,u)$-stable for some
$c, p, u$ (these are not unique).

The property of stability of a matrix may be characterized by the circuits of minimum mean in its associated $\ell$-graph, as shown by the following theorem. This result has been proved by Moln\'arov\'a and Pribi\v{s} \cite{MP00}. Moln\'arov\'a \cite{M05} showed that the same proof is valid for matrices with entries in a divisible Min-Plus algebra.

\begin{theorem}[Theorems 3.1 and 3.4 in \cite{MP00}] \label{carstable}
The length-matrix $\Pi$ of a directed $\ell$-graph $\vec{G}$  
is stable if and only if every non-trivial strongly connected component of $\vec{G}$ contains a circuit of mean $\lambda(\vec{G})$.

Furthermore if  $\Pi$  is stable and $\vec{G}$ contains  circuits, then $\Pi$  is stable with
 period $p$ equals to the lcm of the periodicities of the min-mean components of $\vec{G}$,
and transfer factor $c$ equals to $p \lambda(\vec{G})$. 

\end{theorem}

\begin{corollary} \label{corcarstable}
If the directed $\ell$-graph $\vec{G}$ has at most one non-trivial strongly connected component then its length matrix is stable.
\end{corollary}

\begin{proof} 
Only non-trivial strongly connected components contain circuits of $\vec{G}$ and all vertices of a circuit of $\vec{G}$  belong the same strongly connected component.
Hence if $\vec{G}$ contains at most one non-trivial strongly connected component then the condition of Theorem \ref{carstable} is fulfilled.
\end{proof}

\subsection{Circular strips}

Now we can prove the following theorem.

\begin{theorem} \label{IDstable}
Let $\mathcal G$ be a grid.
For every integer $h \ge 1$, the length-matrix of  $\vec{\mathcal{G}}^{I}_h$ is stable.

\end{theorem}

\begin{proof}
If $\vec{\mathcal{G}}^{I}_h$ contains at most one vertex belonging to a non-trivial strongly connected component, then it contains at most one such component and then 
by Corollary \ref{corcarstable} the length-matrix of  $\vec{\mathcal{G}}^{I}_h$ is stable.

Assume now that there exist two distinct vertices $x$ and $y$ of $\vec{\mathcal{G}}^{I}_h$ that belong each to a non-trivial strongly connected component.
We claim that $x$ and $y$ should then be in the same strongly connected component.
Indeed, consider the labeling $f$ of the strip $\mathcal G_{h,11}$, defined as follows: $f_{1,4}= x$ ;  $f_{5,7}$ labels with $1$  all
the vertices in the 5-th,  6-th and 7-th columns of $\mathcal G_{h,11}$ ; $f_{8,11}=y$.
By definition of a non-trivial strongly connected component, $x$ and $y$ belong each to at least one circuit of $\vec{\mathcal{G}}^{I}_h$. So there exist vertices $x'$ and $y'$ such that $xx'$ 
and $y'y$ are arcs of $\vec{\mathcal{G}}^{I}_h$. 
Since $xx'$ is an arc of $\vec{\mathcal{G}}^{I}_h$, $x$ and $x'$ are $I$-compatible and so $x \triangleright x'$ satisfies $\mathcal P^{I}$. Then either $h \neq 2$ or $\mathcal G \neq \mathcal K$ :
indeed the two vertices of any column of $\mathcal{K}_{2,5}$ have the same closed neighborhood and no code may separate them.
Notice that $f_{1,5}$ is equal to the labeling of $\mathcal G_{h,5}$ obtained from $x \triangleright x'$  by changing all $0$-labels in the last column by a $1$-label.
So, by Lemma~\ref{lem:superset}, it still satisfies $\mathcal P^{I}$.
Consider the labeling $f_{2,6}$ on the columns $2$ to $6$. Since the vertices on columns 3 and 4 are dominated and separated from each other in $x \triangleright x'$ this remains
true in $f_{2,6}$. Furthermore, since all vertices in column $6$ are labeled $1$ by $f$,  and  either $h \neq 2$ or $\mathcal G \neq \mathcal K$, we get that the vertices in column $5$ are dominated and
separated from those in columns $3$ and $4$ and from each other. So $f_{2,6}$ satisfies $\mathcal P^{I}$.
Consider $f_{3,7}$: the vertices on columns $5$ and $6$ are all labeled $1$, so, as $h \neq 2$ or $\mathcal G \neq \mathcal K$, the vertices in columns $4,5$ and $6$ are dominated and separated 
from each other by the labeling $f_{3,7}$.
Consider now $f_{4,8}$: in this labeling all three central columns are completely labeled $1$, so again it satisfies $\mathcal P^{I}$.
By symmetry we get that $f_{5,9}$, $f_{6,10}$, $f_{7,11}$ all satisfy $\mathcal P^{I}$.
Then, $x= f_{1,4}$, $z_i= f_{i,i+3}$ (i= 2, \ldots, 7),
$y= f_{8,11}$, are vertices of $\vec{\mathcal{G}}^{I}_h$
and $xz_2$, $z_2z_3$, $z_3z_4$, $z_4z_5$, $z_5z_6$, $z_6z_7$ and $z_7y$ are arcs of $\vec{\mathcal{G}}^{I}_h$, so that $xz_2z_3z_4z_5z_6z_7y$ is a path from $x$ to $y$ in 
$\vec{\mathcal{G}}^{I}_h$. As $x$ and $y$ were any two vertices in a non-trivial strongly connected component this imply that there exists only one such component of 
$\vec{\mathcal{G}}^{I}_h$.
By Corollary \ref{corcarstable} the proof is done.
\end{proof}

As a corollary of Remark \ref{stabper} and Theorems \ref{chem}, \ref{lemparc} and \ref{IDstable}, for a given $h$, there is a constant-time algorithm to compute the minimum cardinality of an $ID$-code in a circular strip of height $h$ (see the algorithm $Stable$ $(\mathcal{P}, w, M)-$MRP in \cite{BMP}).

\subsection{Non-circular strips}

We consider here the problem of computing the minimum cardinality of an $ID$-code in non-circular strips. It can be solved almost as for the case of circular strip but we have to take into account the specificity of the beginning and end of the strip.

Let $f$ be a $2$-labeling of a strip $\mathcal{G}_{h,4}$ of a grid $\mathcal{G}$. We will say that $f$ satisfies the property {\boldmath $\mathcal P^{I}_b$} ("$b$" for beginning) if the vertices in the first three columns of $\mathcal{G}_{h,4}$ are dominated and separated from each other  by the vertices of $C_f$. 

Similarly $f$ satisfies the property {\boldmath $\mathcal P^{I}_e$} ("$e$" for ending) if the vertices in the last three columns of $\mathcal{G}_{h,4}$ are dominated and separated from each other by the vertices of $C_f$. 

It is easy to see that one can check within a finite number of steps if a labeling of a strip of size $4$ satisfies $P^{I}_b$ or $P^{I}_e$. From the following theorem we can then deduce that being an $ID$-code is a pseudo-5-local property of strips (as defined in \cite{BMP}).

\begin{theorem} \cite{BMP} \label{fivestrip}  The code $C_f$ associated to a $2$-labeling $f$ of a strip $\mathcal{G}_{h,s}$ ($s \ge 4$)  is an $ID$-code of $\mathcal{G}_{h,s}$ if and only $f$ satisfies $\mathcal P^{I}$, $f_{1,4}$ satisfies  
$\mathcal P^{I}_b$ and $f_{s-3,s}$ satisfies  $\mathcal P^{I}_e$.                                                                                                                                                                                                                                                
\end{theorem}

We denote by {\boldmath $\vec{\Gamma}^{I}_h$} the directed $\ell$-graph obtained from $\vec{\mathcal{G}}^{I}_h$ as follows : \begin{itemize}
	\item[-] add two specific vertices: a source $s$ and a sink $t$,
	\item[-] for each vertex $u$ of $\vec{\mathcal{G}}^{I}_h$ such that $u$ satisfies $\mathcal P^{I}_b$, add an arc $su$ of length $\ell(s,u)$ equal to the number of vertices of $\mathcal{G}_{h,4}$ labeled $1$ by $u $,
	\item[-] for each vertex $v$ of $\vec{\mathcal{G}}^{I}_h$ such that $v$ satisfies $\mathcal P^{I}_e$, add an arc $ut$ of length $\ell(u,t)$ equal to $0$.
	\end{itemize}

From Theorem $5$ and Section 5 in \cite{BMP}, valid for any pseudo-$d$-local property, we get the following analog of Theorem \ref{chem} :

\begin{theorem}\cite{BMP} 
 For every  integers $k\ge 4$ and $c \ge 0$, there exists an $ID$-code of $\mathcal{G}_{h,k}$ of cardinality $c$ if and only if $\vec{\Gamma}^{I}_h$  contains a $(k-2)$-path from $s$ to $t$ of length $c$.
\end{theorem}

Furthermore we have again a theorem of stability.

\begin{theorem} \label{IDstable+}
For every integer $h \ge 1$, the length-matrix  of the $\ell$-graph  $\vec{\Gamma}^{I}_h$ is stable.

\end{theorem}

\begin{proof}
The directed graph $\vec{\Gamma}^{I}_h$ is obtained from $\vec{\mathcal{G}}^{I}_h$ by adding appropriately a source $s$ and a sink $t$. These two vertices are trivial strongly connected components of $\vec{\Gamma}^{I}_h$, so the number of non trivial strongly connected components is the same in the two directed graphs. By the proof of Theorem \ref{IDstable} we know that this number is equal to $1$. Corollary \ref{corcarstable} concludes the proof.
\end{proof}

As a corollary we get again that there is a constant-time algorithm that computes the minimum cardinality of an $ID$-code in a non-circular strip of height $h$ (for a fixed $h$).

\subsection{Infinite strips}
In the case on an infinite strip there exists clearly no finite $ID$-code and we need another way to define the size of a "minimum code", using the concept of density. In a finite graph $G=(V,E)$ the {\it density $d_G(C)$} of a code $C$ of $G$ is equal to $\frac{\vert C \vert}{ \vert V \vert}$. We define the {\it density $D(C)$ of a code in the infinite strip $\mathcal G_h$} ($h \ge 1$) as $$D(C) = \limsup\limits_{n \rightarrow +\infty}   \frac{\vert C \cap V_n \vert}{ \vert V_n \vert}$$ where $V_n$ is the set of vertices $(x,y)$ of $\mathcal G_h$ such that $\vert y \vert \le n$ (in other words,  $V_n$ is the set of vertices of $\mathcal G_h$ that belong to the columns numbered from $-n$ to $n$). 

From Lemma \ref{lem:superset} and the fact that in a strip of height $2$ of the king grid, two vertices that are in the same column have exactly the same closed neighborhood, it is easy to deduce the following fact.

\begin{prop} \label{existidcode}
There exists no $ID$-code of a king strip of height $2$. All other strips of size at least $3$ have an $ID$-code.
\end{prop}

We have the following corollary of Theorem \ref{carstable} (see Corollary $1$ in~\cite{BMP}).

\begin{corollary} \label{mindensity}
Let $\mathcal G_h$ be an infinite strip such that $\mathcal G \neq \mathcal K$ or $h \neq 2$ and 
let $\lambda=\lambda(\vec{\mathcal{G}}_h^I)$ be the minimum mean of an elementary circuit of $\vec{\mathcal{G}}_h^I$.
The minimum density on an $ID$-code of $\mathcal G_h$ is
$D_{\mathcal G_h}=  
\frac{\lambda}{h}$.
\end{corollary}

\begin{proof} 
Let $C$ be an identifying code of the infinite strip $\mathcal G_h$ and $f$ the associated labeling of the vertices of $\mathcal G_h$ ($f(v)=1$ if $v \in C$). 
From Proposition~\ref{existidcode},  there exists such a code and the circular strip $\mathcal{G}_{h,3}^{\circ}$ also has one. Then by Theorem \ref{chem} the directed graph $\vec{\mathcal{G}}_h^I$ contains at least one circuit and $\lambda \neq \infty$.

By Theorem \ref{IDstable}, the length matrix $\Pi$ of $\vec{\mathcal{G}}_h^I$ is $(c,p,u)$-stable for some integers $c,p,u$ and by Theorem \ref{carstable} the transfer factor $c$ is equal to $p \lambda$.

We remark that, for every $n \ge 4$,  $f_{-n,n}$ satisfies $\mathcal P^I$, so by Theorem \ref{lemparc},
$\vert  C \cap V_n \vert  \ge Min\{\pi \vert \mbox{ } \pi \mbox{ entry of }  \Pi^{2n+1}\}$. Let $m$ be the minimum entry in the matrices $\Pi^u,\Pi^{u+1},\ldots,\Pi^{u+p-1}$. If $2n+1 \ge u$ then $2n+1 =u+k+jp$ for some integers $0 \le k \le p-1$ and $j\ge 0$ and  we have $\Pi^{2n+1}= \Pi^{u+k} + jp\lambda$, so $Min\{\pi \vert \mbox{ } \pi \mbox{ entry of }  \Pi^{2n+1}\} \ge m +jp\lambda \ge m + (2n+1-u-p+1)\lambda$. 

We get that the density of $C$, $D(C) = \limsup\limits_{n \rightarrow +\infty}   \frac{\vert C \cap V_n \vert}{ \vert V_n \vert} \ge \limsup\limits_{n \rightarrow +\infty} \frac{m}{ h(2n+1)} + \frac{(2n+2-u-p)\lambda}{h(2n+1)}=\frac{\lambda}{h}$. So $D_{\mathcal G_h} \ge \frac{\lambda}{h}$.

Consider now an elementary circuit $\mathcal C$ of $\vec{\mathcal{G}}_h^I$ of mean $\lambda$ (by assumption there exists  at least one such circuit) and let $k$ be the cardinality of $\mathcal C$. This circuit corresponds to a labeling $f^*$ of the strip $\mathcal G_{h,k+4}$ that satisfy $\mathcal P^I$ and such that $f^*_{1,4}=f^*_{k-3,4}$. The density of $C_{f^*_{5,k}}$ on the last $k$ columns of $\mathcal G_{h,k+4}$ is equal to $\frac{k \lambda}{kh}=\frac{\lambda}{h}$. The code of the infinite strip corresponding to an infinite repetition of $f^*_{5,k}$ is an $ID$-code of $\mathcal G_h$ of density equal to $\frac{\lambda}{h}$.

Thus we have proved that $D_{\mathcal G_h}= \frac{\lambda}{h}$.
\end{proof}

Remark that, by Corollary \ref{mindensity}, the problem of computing the minimum density of an $ID$-code of an infinite strip $S$ of height $h$ is the same as the problem of computing the minimum cardinality of an $ID$-code of a circular strip of height $h$ on the same grid than $S$.

In the next section, we explain how we implemented the algorithm we have described above to get $ID$-codes of minimum density for circular and infinite strips of grids of height at most $4$.

\section{Implementation of the algorithm}\label{sec:algo}

\subsection{General scheme}\label{sec:implementation}

The algorithms were implemented using the C++ language. They were designed to be executed in multithread, that is to say on several processors in parallel. These algorithms were run on the computational server of the G-SCOP lab having 10 processors.

The first task consisted in generating all possible $2$-labelings of strips of a given height and size $4$ (vertices of the   graph) and the entries of the length-matrix (lengths of the arcs of the auxiliary graph). Coefficients of powers of the length-matrix matrix $\Pi$ of the auxiliary graph $\vec{\mathcal{G}}$ were stored as 16-bits \textit{shorts}.

At any step $k$ of the algorithm, in order to compute $\Pi^k$, three matrices need to be stored in the RAM: The initial length matrix $\Pi$, its power $\Pi^{k-1}$, as well as its power $\Pi^{k}$ that we compute as the product of $\Pi$ with $\Pi^{k-1}$.

In order to detect a period in the sequence of matrices, we need to store on the hard disk drive the matrix obtained at each step $k$. 

When a period is detected, we get the values $c, p, u$ such that $\Pi$ is $(c,p,u)$-stable. This enables one to find the minimum cardinality of a code in a strip of size $n$ using only a constant number of elementary operations.

If one wishes to obtain also the configuration of an optimal code, then one can perform a backtrack analysis of the algorithm, in order to get an optimal circuit of the auxiliary graph with the desired number of arcs.

\subsection{Technical tricks to speed up the process}

\paragraph{Size of the matrices}

The number of vertices, hence the size of the length-matrix $\Pi$, increases rapidly as the height of the strip increases. For instance, in the case of $ID$-codes, for the strip of the square grid of height $3$, the auxiliary graph has $16\,824$ vertices. Using the approach described in Section~\ref{sec:implementation}, the size of a power of $\Pi$ is approximately 540 Mo. Hence, in this algorithm, the size of the matrices is a critical parameter, since we have to be able to store three such matrices in the RAM.

\paragraph{Detecting the period}

In order to detect a period in the sequence of matrices, we stored on the hard disk drive the matrix $\widetilde{\Pi}_k=\Pi^k-\min_{i,j}(\Pi^k_{i,j})$ for each $k$, instead of $\Pi^k$. Hence there is a period when we find   $k' > k$ such that $\widetilde{\Pi}_k = \widetilde{\Pi}_{k'}$. In oder to speed up the process, hashcodes of each matrix were computed. Since different values of the hashcode ensure that the matrices are different, this enables one to avoid a large number of tests of the form \textit{``do we have $\widetilde{\Pi}_k = \widetilde{\Pi}_{k'}$ ?''}.

\paragraph{Speeding up the backtrack}

Due to the prohibitive size of the matrices, we did not perform any backtrack to get optimal codes. Indeed, a backtrack would have required to load into the RAM each of the matrices computed before the detection of the pseudo-period. Instead, we used constraint programming, using the java language and the \textsc{Choco} library. On a personal computer, the program finds an optimal code in less than 1 second for strips of height 1 and 2. For height 3, the computation time is about 1 hour. For height 4, the computation time is about 1 day.

\subsection{Running times}

We provide here the running times and the size of the length-matrixmatrix for the case of $ID$-codes in the strip of the square grid. Running times in strips of other types of grids are of the same order.

\begin{table}[ht!]
\begin{tabular}{|c|ccc|}
\hline
Height & Number of vertices & Computation time & Size of a matrix\\
\hline
\hline
1 & 10 & 1 sec & 200 o\\
\hline
2 & 169 & 2 sec & 56 Ko\\
\hline
3 & 2\,598 & 6 min & 13 Mo\\
\hline
4 & 37\,791 & 16 days & 2,6 Go\\
\hline
\end{tabular}
\caption{Running times and matrix size in the case of $ID$-codes in strips of the square grid.}
\end{table}

\section{New results on $ID$-, $LD$-, and $LTD$-codes in finite circular strips and infinite strips}\label{sec:results}

In this section, we report the results we obtained thanks to our implementation of the algorithm described above (and  in \cite{BMP}) for computing the minimum cardinality of an $ID$-, $LD$-, or $LTD$-code in finite circular strips and the minimum density of  such codes in infinite strips.

We remark that the number of vertices in a strip of size $n$ and height $h$ is equal to $nh$ (it is the same for all kind of strips). 

For each case we will underline the period $p$, the transfer factor $c$, the minimum mean $\lambda = \frac{c}{p}$ of a circuit in the auxiliary graph and specify the smaller size  of a circular strip for which the corresponding minimum density $\frac{\lambda}{h}$ is attained as well as one corresponding pattern. By Corollary \ref{mindensity} a code of minimum density of the infinite strip is obtained by an infinite repetition of such a pattern.

The strips of height $1$ of the king grid, and of the triangular grid are the same as the one of the square grid, so this case is studied only in the square grid section.
Toroidal grid or strips are defined only for an height at least $3$.

\subsection{Identifying codes}

\renewcommand{\arraystretch}{1.5}

\subsubsection{Square grid}

Some of the results stated here for strips of height $1$ or $2$ were already  in \cite{BCHL04} and in \cite{DGM04}.

\begin{prop} 
Let $ID^{\mathcal S}(n, h)$ denote the minimum cardinality of an $ID$-code in a circular strip of the square grid of size $n$ and height $h$:

\begin{itemize}

	\item  $h=1$: 
	$\displaystyle ID^{\mathcal S}(n, 1)= \left\{\begin{array}{ll}
	3, \mbox{ for } n=5\\
	\frac{n}{2}, \mbox{ for } n \ge 6 \mbox{ and } n \equiv 0 [2] \\
\lceil\frac{n}{2}\rceil+1, \mbox{ for } n \ge 7 \mbox{ and } n \equiv 1 [2].\\
\end{array}\right.$

So that, $p=2$, $c=1$, and $\lambda = \frac{1}{2}$ is the minimum density (see Figure~\ref{bandecarrecode} for a pattern of minimum density that applies for any circular strip of even size greater than or equal to $6$)

	\item $h=2$: 
	$\displaystyle ID^{\mathcal S}(n, 2)= \left\{\begin{array}{ll}
\lceil\frac{6n}{7}\rceil+1, \mbox{ for } n \ge 8 \mbox{ and } n \equiv 1 \mbox{ or } 2 [7] \\
\lceil\frac{6n}{7}\rceil, \mbox{ for } n \ge 5, \mbox{ and } n \equiv 0,3,4,5 \mbox{ or } 6 [7]. 
\end{array}\right.$

So that, $p=7, c=6$,  $\lambda =\frac{6}{7}$ corresponds to the minimum density ${{\frac{6}{7}}/{2}}=\frac{3}{7}$ (see Figure~\ref{bandecarrecode} for a pattern of minimum density that applies for any circular strip whose size is a multiple of $7$).

	\item $h=3$: 
	$\displaystyle ID^{\mathcal S}(n, 3)= \left\{\begin{array}{ll}
	
\lceil\frac{7n}{6}\rceil, \mbox{ for } n \ge 5 \mbox{ and } n \equiv 0,1,2,3,4,5,7,8,9, \mbox{or } 10 [12] \\
\lceil\frac{7n}{6}\rceil+1, \mbox{ for } n \ge 6 \mbox{ and } n \equiv 6, \mbox{or } 11 [12]. \\
\end{array}\right.$

So that, $p=12, c=14$,  $\lambda =\frac{7}{6}$ corresponds to the minimum density $\frac{7}{6}/{3}=\frac{7}{18}$ (see Figure~\ref{bandecarrecode} for a pattern of minimum density that applies for any circular strip whose size is a multiple of $12$).

\item $h=4$: 

	$\displaystyle ID^{\mathcal S}(n, 4)= \left\{\begin{array}{ll}
	\frac{11n}{7}, \mbox{ for } n \ge 14 \mbox{ and } n \equiv 0 [14] \\
\lceil\frac{11n}{7}\rceil, \mbox{ for } n \ge 5 \mbox{ and } n \equiv 1,2,3,4,5, \mbox{or } 6 [7] \\
\frac{11n}{7}+1, \mbox{ for } n \ge 7 \mbox{ and } n \equiv 7 [14]. \\
\end{array}\right.$

So that, $p=14, c=22$,  $\lambda =\frac{11}{7}$ corresponds to the minimum density $\frac{11}{7}/{4}=\frac{11}{28}$ (see Figure~\ref{bandecarrecode} for a pattern of minimum density that applies for any circular strip whose size is a multiple of $14$).

\end{itemize}
\end{prop}

\begin{figure}[ht!]
	\centering
	\includegraphics[width=120mm]{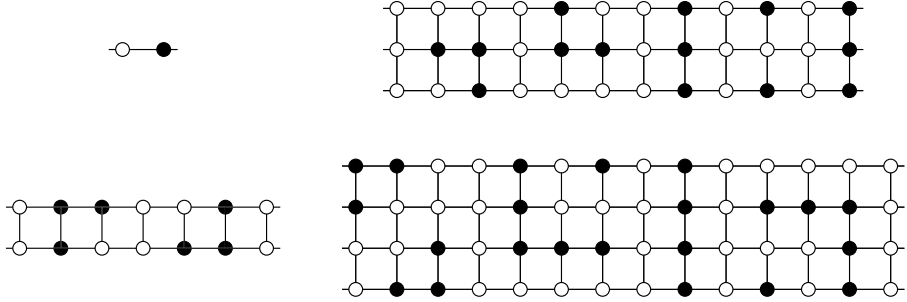}
	\caption{Periodic patterns for minimum density $ID$-codes of infinite square strips of heights 1, 2, 3, 4.}
	\label{bandecarrecode}
\end{figure}

\subsubsection{King grid}

\begin{prop} 
Let $ID^{\mathcal K}(n, 3)$ denote the minimum cardinality of an $ID$-code in a circular strip of the king grid of size $n \ge 5$ and height $3$:

	$\displaystyle ID^{\mathcal K}(n, 3)= \left\{\begin{array}{ll}
	n+1, \mbox{ for } n=7,9,13,19\\
	n, \mbox{ for }  n \neq 7,9,13,19.  \\
\end{array}\right.$

So that, $p=c=1$,  $\lambda =1$ corresponds to the minimum density $\frac{1}{3}$ (see Figure~\ref{bandeautrecode} for a pattern of minimum density that applies for any circular strip of even size at least $6$).

\end{prop}

\subsubsection{Toroidal circular strip}

\begin{prop} 
Let $ID^{\mathcal S^T}(n, h)$ denote the minimum cardinality of an $ID$-code in a toroidal circular strip of size $n \ge 5$ and height $h$.

\begin{itemize}

\item $h=3$:
$\displaystyle ID^{\mathcal S^T}(n, 3) =~\lceil\frac{5n}{4}\rceil$.

	So that, $p=4, c=5$,  $\lambda =\frac{5}{4}$ corresponds to the minimum density $\frac{5}{4}/{3}=\frac{5}{12}$ (see Figure~\ref{bandeautrecode} for a pattern of minimum density that applies for any circular strip whose size is a multiple of $4$).

\item $h=4$:
$\displaystyle ID^{\mathcal S^T}(n, 4) =
\left\{\begin{array}{ll}
~\lceil\frac{10n}{7}\rceil+1  \mbox{ for } n=7,9,14,16,21,35,63\\
	\lceil\frac{10n}{7}\rceil, \mbox{ for } n \neq 7,9,14,16,21,35,63.  \\
	
	\end{array}\right.$
	
		So that, $p=7, c=10$,  $\lambda =\frac{10}{7}$ corresponds to the minimum density $\frac{10}{7}/{4}=\frac{5}{14}$ (see Figure~\ref{bandetor4} for a pattern of minimum density that applies for any circular strip whose size is a multiple of $28$).

\end{itemize}

\end{prop}

\subsubsection{Triangular grid}

\begin{prop} 
Let $ID^{\mathcal T}(n, h)$ denote the minimum cardinality of an $ID$-~code in a circular strip of the triangular grid of size $n \ge 5$ and height $h$.

\begin{itemize}
	
	\item $h=2$: 
	$\displaystyle ID^{\mathcal T}(n, 2)= 	n.$
	
	So that, $p=c=1$ and  $\lambda = 1$ corresponds to the minimum density $\frac{1}{2}$  (see Figure~\ref{bandeautrecode} for a pattern of minimum density that applies for any circular strip of size at least $4$).

	\item $h=3$: 
	$\displaystyle ID^{\mathcal T}(n, 3)= \left\{\begin{array}{ll}
	n+1,  \mbox{ for } n=7  \\
	n, \mbox{ for } n \neq 7.  \\
	
	\end{array}\right.$

	So that, $p=c=1$ and $\lambda = 1$ corresponds to the minimum density $\frac{1}{3}$ (see Figure~\ref{bandeautrecode} for a pattern of minimum density that applies for any circular strip of even size at least $6$).

\end{itemize}

\end{prop}

\begin{figure}[ht!]
	\centering
	\includegraphics[width=50mm]{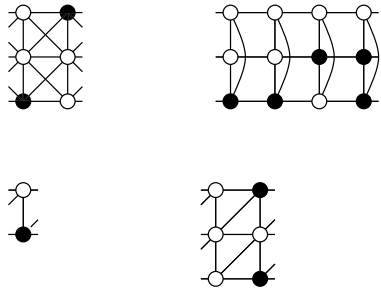}
	\caption{Periodic patterns for minimum density $ID$-codes of infinite king, triangular and toroidal strips of heights 2, 3.}
	\label{bandeautrecode}
\end{figure}

\begin{figure}[ht!]
	\centering
	\includegraphics[width=0.9\textwidth]{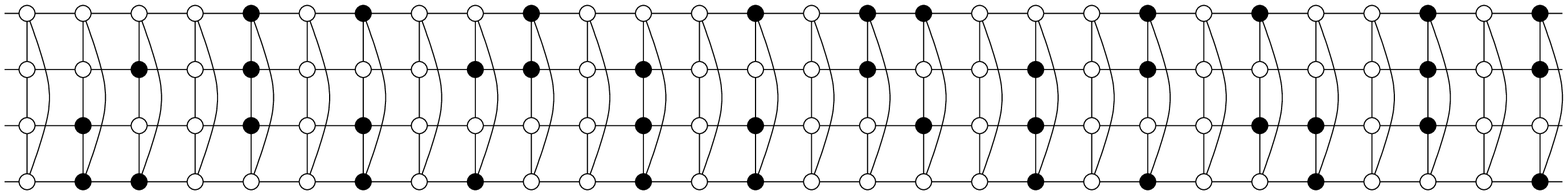}
	\caption{Periodic pattern for minimum density $ID$-codes of the toroidal strips of height 4.}
	\label{bandetor4}
\end{figure}

\subsection{Locating-dominating codes}

\subsubsection{Square grid}

\begin{prop} 
Let $LD^{\mathcal S}(n, h)$ denote the minimum cardinality of an $LD$-code in a circular strip of the square grid of size $n \ge 5$ and height $h$:

\begin{itemize}

	\item  $h=1$: 	
	$LD^{\mathcal S}(n, 1)= \lceil\frac{2n}{5}\rceil.$

So that, $p=5$, $c=2$, $\lambda = \frac{2}{5}$ is the minimum density (see Figure~\ref{bandecarreld} for a pattern of minimum density that applies for any circular strip whose size is a multiple of $5$).
This result was already stated in \cite{Slater88}.

\item $h=2$: 
	$\displaystyle LD^{\mathcal S}(n, 2)= \left\{\begin{array}{ll}
		\lceil\frac{3n}{4}\rceil, \mbox{ for }  n \equiv 0,1,2,3,5,6 \mbox{ or } 7 [8] \\
	\frac{3n}{4}+1, \mbox{ for } n \equiv 4 [8]. \\

	\end{array}\right.$

So that, $p=8, c=6$, $\lambda = \frac{3}{4}$ corresponds to the minimum density $\frac{3}{4}/{2}=\frac{3}{8}$ (see Figure~\ref{bandecarreld} for a pattern of minimum density that applies for any circular strip whose size is a multiple of $8$).

	\item $h=3$: 
	$\displaystyle LD^{\mathcal S}(n, 3)= \left\{\begin{array}{ll}
	n, \mbox{ for }  n \equiv 0,2,3 \mbox{ or } 4 [6] \\
	n+1, \mbox{ for }  n \equiv 1 \mbox{ or } 5 [6]. \\
\end{array}\right.$

So that, $p=6, c=6$, $\lambda = 1$ corresponds to the minimum density $\frac{1}{3}$  (see Figure~\ref{bandecarreld} for a pattern of minimum density that applies for any circular strip of even size at least $6$).
\end{itemize}
\end{prop} 

\begin{figure}[ht!]
	\centering
	\includegraphics[width=70mm]{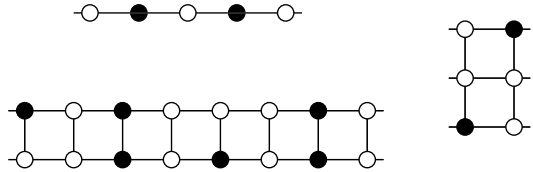}
	\caption{Periodic patterns for minimum density $LD$-codes of infinite square strips of heights 1, 2, 3.}
	\label{bandecarreld}
\end{figure}

\subsubsection{King grid}

\begin{prop} 
Let $LD^{\mathcal K}(n, h)$ denote the minimum cardinality of an $LD$-code in a circular strip of the king grid of size $n \ge 5$ and height $h$:

\begin{itemize}

	\item  $h=2$: 	$\displaystyle LD^{\mathcal K}(n, 2)= 	n.$

So that, $p=c=1$, $\lambda = 1$ corresponds to the minimum density $\frac{1}{2}$  (see Figure~\ref{bandeautreld} for a pattern of minimum density that applies for any circular strip of size at least $4$).

\item $h=3$: 
	
	$ LD^{\mathcal K}(n, 3)= \lceil\frac{4n}{5}\rceil$.

	So that, $p=5$, $c=4$, $\lambda = \frac{4}{5}$ corresponds to minimum density $\frac{4}{5}/{3}=\frac{4}{15}$  (see Figure~\ref{bandeautreld} for a pattern of minimum density that applies for any circular strip whose size is a multiple of $5$).

	\end{itemize}

\end{prop}

\subsubsection{Toroidal grid}

\begin{prop} 
Let $LD^{\mathcal S^T}(n, 3)$ denote the minimum cardinality of an $LD$-code in a toroidal circular strip of  of size $n \ge 5$ and height $3$:
$\displaystyle LD^{\mathcal S^T}(n, 3)=~n.$
	
	So that, $p= c=1$, $\lambda = 1$ corresponds to the minimum density $\frac{1}{3}$  (see Figure~\ref{bandeautreld} for a pattern of minimum density that applies for any circular strip of even size at least $4$). 

	\end{prop} 

\subsubsection{Triangular grid}

\begin{prop} 
Let $LD^{\mathcal T}(n, h)$ denote the minimum cardinality of an $LD$-code in a circular strip of the triangular grid of size $n \ge 5$ and height $h$.

\begin{itemize}
	
	\item $h=2$:  $LD^{\mathcal T}(n, 2)= \lceil\frac{2n}{3}\rceil$.
	
		So that, $p=3, c=2$, $\lambda = \frac{2}{3}$ corresponds to the minimum density $\frac{2}{3}/{2}=\frac{1}{3}$  (see Figure~\ref{bandeautreld} for a pattern of minimum density that applies for any circular strip whose size is a multiple of $3$ greater than or equal to $6$).

	\item $h=3$: 
	$LD^{\mathcal T}(n, 3)= \lceil\frac{9n}{10}\rceil$.

So that, $p=10, c=9$, $\lambda = \frac{9}{10}$ corresponds to the minimum density $\frac{9}{10}/{3}=\frac{3}{10}$  (see Figure~\ref{bandeautreld} for a pattern of minimum density that applies for any circular strip whose size is a multiple of $10$).

\end{itemize}

\end{prop}

\begin{figure}[ht!]
	\centering
	\includegraphics[width=60mm]{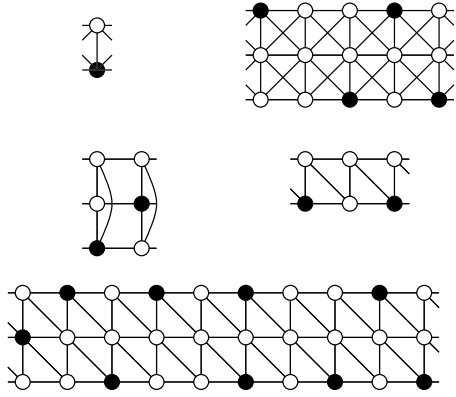}
	\caption{Periodic patterns for minimum density $LD$-codes of infinite king, triangular and toroidal strips of heights 2, 3.}
	\label{bandeautreld}
\end{figure}

\subsection{Locating-total-dominating codes}
\subsubsection{Square grid}
\begin{prop} 
Let $LTD^{\mathcal S}(n, h)$ denote the minimum cardinality of an $LTD$-code in a circular strip of the square grid of size $n \ge 5$ and height $h$:

\begin{itemize}

	\item  $h=1$: 
	$\displaystyle LTD^{\mathcal S}(n, 1)= \left\{\begin{array}{ll}
	\lceil\frac{n}{2}\rceil, \mbox{ for } n \ge 4 \mbox{ and } n \equiv 0,1 \mbox{ or } 3 [4] \\
	\frac{n}{2}+1, \mbox{ for } n \ge 6 \mbox{ and } n \equiv 2 [4]. \\
	
\end{array}\right.$

So that, $p=4$, $c=2$, $\lambda = \frac{1}{2}$ corresponds to the minimum density  (see Figure~\ref{bandecarreltd} for a pattern of minimum density that applies for any circular strip whose size is a multiple of $4$).

	\item $h=2$: 
	$\displaystyle LTD^{\mathcal S}(n, 2)= \left\{\begin{array}{ll}
	6, \mbox{ for } n =6 \\
	\lceil\frac{4n}{5}\rceil, \mbox{ for } n \neq 6.  \\
	\end{array}\right.$

So that, $p=5, c=4$, $\lambda = \frac{4}{5}$ corresponds to the minimum density $\frac{4}{5}/{2}=\frac{2}{5}$  (see Figure~\ref{bandecarreltd} for a pattern of minimum density that applies for any circular strip whose size is a multiple of $5$).

	\item $h=3$: 
	$ LTD^{\mathcal S}(n, 3)= \lceil\frac{7n}{6}\rceil$

So that, $p=6, c=7$, $\lambda = \frac{7}{6}$  corresponds to the minimum density $\frac{7}{6}/{3}=\frac{7}{18}$  (see Figure~\ref{bandecarreltd} for a pattern of minimum density that applies for any circular strip whose size is a multiple of $6$).

\end{itemize}
\end{prop} 

\begin{figure}[ht!]
	\centering
	\includegraphics[width=70mm]{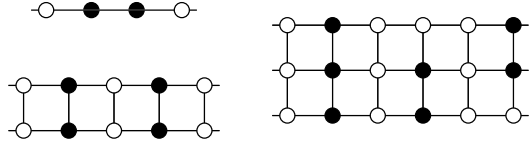}
	\caption{Periodic patterns for minimum density $LTD$-codes of infinite square strips of heights 1, 2, 3.}
	\label{bandecarreltd}
\end{figure}

\subsubsection{King grid}

\begin{prop} 
Let $LTD^{\mathcal K}(n, h)$ denote the minimum cardinality of an $LTD$-code in a circular strip of the king grid of size $n \ge 5$ and height $h$:

\begin{itemize}

	\item  $h=2$: 
	$\displaystyle LTD^{\mathcal K}(n, 2)= 	n$.

So that, $p=c=1$, $\lambda = 1$ corresponds to the minimum density $\frac{1}{2}$  (see Figure~\ref{bandeautreltd} for a pattern of minimum density that applies for any circular strip of size at least $4$).

\item $h=3$: 
	$ LTD^{\mathcal K}(n, 3)= 
		\lceil\frac{8n}{9}\rceil.$

	So that, $p=9$, $c=8$, $\lambda = \frac{8}{9}$ corresponds to the minimum density $\frac{8}{9}/{3}=\frac{8}{27}$  (see Figure~\ref{bandeautreltd} for a pattern of minimum density that applies for any circular strip whose size is a multiple of $9$).

	\end{itemize}

\end{prop} 

\subsubsection{Toroidal grid}

\begin{prop} 
The minimum cardinality of an $LTD$-code in a toroidal circular strip of  of size $n \ge 5$ and height $3$ is:
 
$\displaystyle  LTD^{\mathcal S^T}(n, 3)= \left\{\begin{array}{ll}
	n,  \mbox{ for } n \equiv 0 [6] \\
	n+1, \mbox{ for }  n \mbox{ } \cancel{\equiv} \mbox{ } 0 [6].\\
	
	\end{array}\right.$

So that, $p=c=1$, $\lambda =1$ corresponds to the minimum density $\frac{1}{3}$ (see Figure~\ref{bandeautreltd} for a pattern of minimum density that applies for any circular strip whose size is a multiple of $6$).
	\end{prop} 

\subsubsection{Triangular grid}

\begin{prop} 
The minimum cardinality of an $LTD$-code in a circular strip of the triangular grid of size $n\ge 5$ and height $h$ is:

\begin{itemize}

	\item  for $h=2$: 
	$ LTD^{\mathcal T}(n, 2)= 
	\lceil\frac{2n}{3}\rceil.$

So that, $p=3$, $c=2$, $\lambda = \frac{2}{3}$ corresponds to the minimum density $\frac{2}{3}/{2}=\frac{1}{3}$ (achieved for any circular strip whose size is a multiple of $3$ greater than or equal to $6$, see Figure~\ref{bandeautreltd}).

\item for $h=3$: 
	$\displaystyle LTD^{\mathcal T}(n, 3)=n.$
	
	So that, $p=c=1$, $\lambda = 1$ corresponds to the minimum density $\frac{1}{3}$ achieved for any circular strip of size at least $5$, (see Figure~\ref{bandeautreltd} for a pattern valid for any circular strip whose size is a multiple of $3$ greater than or equal to $6$, see Figure~\ref{bandeautreltd}).

		\end{itemize}

\end{prop}

\begin{figure}[ht!]
	\centering
	\includegraphics[width=90mm]{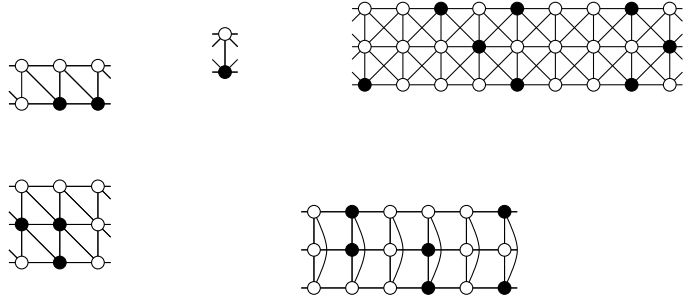}
	\caption{Periodic patterns for minimum density $LTD$-codes of infinite king, triangular and toroidal strips of heights 2, 3.}
	\label{bandeautreltd}
\end{figure}

\subsection{Infinite strips}

By Corollary \ref{mindensity}, the optimal results on circular strips provide those for the infinite strips. The results we obtained for infinite strips of height at most $4$ are summarized in Table~\ref{tablresum} and corresponding patterns are in Figures~\ref{bandecarrecode}--\ref{bandeautreltd}. All were already stated in \cite{B12}. Notice that the already known values were only for strips of height at most $2$. Until now the best known density for an $ID$-code of a  infinite square strip of height $3$ was $\frac{2}{5}$ \cite{DGM04} and we have shown that the optimal density for such a strip is $\frac{7}{18}$. The same result is proved by a Discharging Method in \cite{BHP}. Similarly the best density of an $LTD$-code in an infinite square strip has also been proved by Junnila  \cite{J} by using a "Share Method" as defined by Slater \cite{Slater02}. By a computer search similar to ours, Jiang~\cite{J16} recently and independently found the optimal density of an indentifying code in infinite square strips of heights~4 and~5. In Table~\ref{tablresum}  we give also all known results on the minimum density in infinite grids (there, when there are two references, the first  is for the lower bound, and the second contains the optimal corresponding pattern(s)). Not surprisingly, these minima are lower than those in infinite strips. For identifying codes, the minimum density in an infinite square grid is $\frac{7}{20}$ \cite{BHL05} 
\cite{CGHLMPZ99}, and  in \cite{BHP} the following bounds on the minimum density of an identifying code of an infinite square strip are proved: 
$$\displaystyle \frac{7}{20} + \frac{1}{20h} \leq ID^{\mathcal S}(\infty, h) \leq \min \left\{\frac{2}{5}, \frac{7}{20} + \frac{3}{10h} \right\}.$$
Note that, however,  the smallest density of an identifying code in an infinite square strip is lower in the case of  height $3$ ($\frac{7}{18}$) than in the case of height~$4$ ($\frac{11}{28}$).

\begin{table}[ht!]
\begin{changemargin}{-15mm}{15mm}{\renewcommand{\arraystretch}{2.2}
\begin{tabular}{|c||c|c|c||c|c|c||c|c|c||c|c|c|}
	 \hline  & \multicolumn{3}{|c||}{ $\mathcal{S}_h$ } & \multicolumn{3}{|c ||}{ $\mathcal{K}_h$} & \multicolumn{3}{|c||}{ $\mathcal{S}_{\circ h}$ } & \multicolumn{3}{|c|}{ $\mathcal{T}_h$ }  \\
\hline
$h$	& ID & LD & LTD & ID & LD & LTD &  ID & LD & LTD & ID & LD & LTD   \\ \hline
	1 & $\frac{1}{2}$  & $\frac{2}{5}$  & $\frac{1}{2}$ & X & X & X & X & X & X & X & X & X \\
		& \cite{BCHL04} & \cite{BCHL04} & \cite{HHH06} & & & & & & & & & \\ \hline
	2 &  $\frac{3}{7}$ & $\frac{3}{8}$ & $\frac{2}{5}$ & $\emptyset$ & $\frac{1}{2}$ & $\frac{1}{2}$ & X & X & X & $\frac{1}{2}$ &  $\frac{1}{3}$ & $\frac{1}{3}$\\
	& \cite{DGM04} & & \cite{HHH06}& & & & & & & & & \\ \hline
	3 & $\frac{7}{18}$ & $\frac{1}{3}$ & $\frac{7}{18}$ & $\frac{1}{3}$ & $\frac{4}{15}$ & $\frac{8}{27}$ & $\frac{5}{12}$ & $\frac{1}{3}$ & $\frac{1}{3}$ & $\frac{1}{3}$ & $\frac{3}{10}$ & $\frac{1}{3}$ \\ &  \cite{BHP} &  &  \cite{J} & & & & & & & & & \\ \hline 
4	& $\frac{11}{28}$ & & & & & & $\frac{5}{14}$ & & & & & \\ \hline
$\infty$	& $\frac {7}{20}$  & $\frac{3}{10}$ & &$\frac {2}{9}$ &$\frac{1}{5}$ & & X & X & X &$\frac{1}{4}$ &$\frac{13}{57}$ & \\  & \cite{BHL05} 
\cite{CGHLMPZ99}   
& \cite{Slater02} & &  \cite{CHLZ01} \cite{CHL02} &    \cite{HL06}   & &  & & &  \cite{KCL98} & \cite{H06}& \\ \hline
\end{tabular}\label{tablresum}}\end{changemargin}

\caption{{\bf Minimum densities of codes in infinite strips of height at most $4$ (computed by our algorithm), and minimum densities of codes in infinite grids. }
 A cross ``X'' indicates that the corresponding graph is not relevant. For instance, the graph $\mathcal{K}^\infty_1$ is identical to $\mathcal{S}^\infty_1$. The symbol $\emptyset$ means that the code does not exist for the corresponding graph (the graph $\mathcal{K}^\infty_2$ has no ID code, since, for instance, vertices $(1,1)$ and $(2,1)$ have the same closed neighborhood). Empty cells in the row of the height $4$ correspond to cases for which we did not run the computer search. 
}
\end{table}

\newpage

\bibliographystyle{plain}

\end{document}